\documentclass[letterpaper, 10 pt, journal, twoside]{IEEEtran}

\usepackage{times}
\usepackage{setspace}
\usepackage{url}
\spacing{1}
\usepackage[utf8]{inputenc}
\usepackage[T1]{fontenc}
\usepackage{graphicx}		
\usepackage{wrapfig}
\usepackage{sidecap} 
\usepackage[export]{adjustbox}
\usepackage{float}

\usepackage{amsmath} 
\usepackage{amssymb}  
\usepackage{amsthm}
\usepackage{mathtools}
\usepackage[normalem]{ulem}
\usepackage{paralist}	
\usepackage[space]{grffile} 
\usepackage{color}
\let\labelindent\relax
\usepackage{enumitem}
\usepackage{bm}
\usepackage{cancel}
\usepackage{hhline}

\newtheorem{theorem}{Theorem}

\newtheorem{lemma}{Lemma}
\theoremstyle{definition}

\theoremstyle{remark}
\newtheorem{remark}{Remark}
\theoremstyle{definition}
\newtheorem{assumption}{Assumption}
\theoremstyle{definition}

\newtheorem{example}{Example}

\newcommand{\R}{\mathbb{R}}
\newcommand{\C}{\mathcal{C}}

\definecolor{blue}{RGB}{38,38,134}
\definecolor{darkblue}{RGB}{0,0,102}
\definecolor{lightblue}{RGB}{77,77,148}

\definecolor{gold}{RGB}{234, 170, 0}
\definecolor{metallic_gold}{RGB}{139, 111, 78}

\renewcommand{\cal}[1]{\mathcal{ #1 }}

\newcommand{\der}[2]{\frac{\mathrm{d} #1 }{\mathrm{d} #2 }}
\newcommand{\derp}[2]{\frac{\partial #1 }{\partial #2 }}

\usepackage{xfrac}

\usepackage{dblfloatfix}

\usepackage{cite}
\usepackage[dvipsnames,table,xcdraw]{xcolor}

\usepackage{amsfonts}
\usepackage{algorithmic}
\usepackage{textcomp}

\usepackage{subfigure}
\usepackage{epsfig} 
\newcommand{\new}[1]{{\color{black} #1}}

\begin{document}

\title{ \bf
Disturbance Observers for Robust Safety-critical \\ Control with Control Barrier Functions
}

\author{Anil Alan$^{1}$, Tamas G. Molnar$^{2}$, Ersin Da\c{s}$^{2}$, Aaron D. Ames$^{2}$, G\'abor  Orosz$^{1}$
\thanks{\new{*This research is supported in part by the National Science Foundation, CPS Award \#1932091, Dow (\#227027AT) and Aerovironment.}}%
\thanks{$^{1}$A. Alan and G. Orosz are with the University of Michigan, Ann Arbor, MI 48109, USA. ${\tt\small \{anilalan, orosz \}@umich.edu}$}%
\thanks{$^{2}$T. G. Molnar, E. Da\c{s} and A. D. Ames are with the California Institute of Technology, Pasadena, CA 91125, USA. ${\tt\small \{tmolnar, ersindas, ames\}@caltech.edu}$ } %
}

\maketitle
\thispagestyle{empty}         

\begin{abstract}
This work provides formal safety guarantees for control systems with disturbance.
A disturbance observer-based robust safety-critical controller is proposed, that estimates the effect of the disturbance on safety and utilizes this estimate with control barrier functions to attain provably safe dynamic behavior. 
The observer error bound -- which consists of transient and steady-state parts -- is quantified, and the system is endowed with robustness against this error via the proposed controller.
A \new{connected} cruise control problem is used as illustrative example through simulations including real disturbance data.
\end{abstract}

\section{Introduction} \label{sec:intro}

Safety-critical control has become increasingly crucial for deploying ubiquitous autonomous systems in a priori unknown operational environments.
Examples include robotics and automotive systems, where maintaining safety with control is of utmost priority, even under uncertain dynamics.
Control barrier functions (CBFs) have shown success in achieving this, by providing formal safety guarantees through forward invariance of a pre-defined safe set~ \cite{ames2017control}. 
In particular, CBF-based quadratic programs (CBF-QPs) provide effective solutions for control-affine nonlinear systems and have been implemented in many application domains \cite{nguyen2015safety,breeden2021, thyri2020reactive}.

Many studies on CBFs rely on precise knowledge of the underlying system dynamics.
However, in the presence of model uncertainty or external disturbances, safety guarantees established by CBFs degrade or alter. 
To remedy this concern, robust extensions of CBFs have been proposed that utilize the available knowledge or assumptions about the unmodeled dynamics.
Worst-case uncertainty bounds were incorporated into CBF conditions in \cite{jankovic2018robust, nguyen2021robust} 
to overcome uncertainties. 
This approach may yield conservative results, \new{as will be shown}.
\new{Alternatively}, input-to-state safety (ISSf) characterizes how the safe set changes with disturbances.
It mitigates the conservativeness by bounding safety degradation\new{. However, even ISSf-based methods} may suffer from \new{significant} uncertainty bounds \cite{alan2021safe, Alan__arxiv:22}. \new{While less conservative adaptive control approaches have been proposed to tackle structured parametric uncertainties \cite{lopez2020robust}, their safety guarantees do not include time-varying external disturbances.}

Disturbance observer (DOB) theory -- a robust control technique for suppressing the effects of disturbance and model uncertainty by the feedback of their estimations \cite{chen2015disturbance} -- has recently been adopted in synthesizing safety-critical controllers \cite{das_dist:22}.
The resulting DOB-based scheme estimates the effects of the disturbance on the time derivative of the CBF with an exponentially decaying error bound. 
\new{However, as will be shown, this method may be conservative initially, since it cancels the transient observer error regardless of the initial condition. 
Another DOB-based approach 
observes external disturbances that occur in the system dynamics in an affine expression, multiplied by a known coefficient \cite{Wang__DOB-CBF:22}. Although this method can be effective for the affine problem setup, we seek to ensure robust safety for a more general uncertainty description.}
\new{Additionally, none of these methods consider} how the choice of disturbance observer parameters affects the closed-loop behavior or performance.

To this end, this paper proposes a novel DOB-based safety-critical control framework with CBFs to guarantee robustness against \new{uncertainties}, together with guidelines on the design of DOB and controller parameters.
Our approach takes advantage of the input-to-state stability of the high-gain first-order DOB dynamics introduced in \cite{das_dist:22} and leverages the idea of input-to-state safety~\cite{alan2021safe} to provide robustness against the observer error.
The end result is less conservative robust safety guarantee in the presence of \new{model uncertainties.}

\begin{figure}[t]
	\centering
	\includegraphics[scale=0.77]{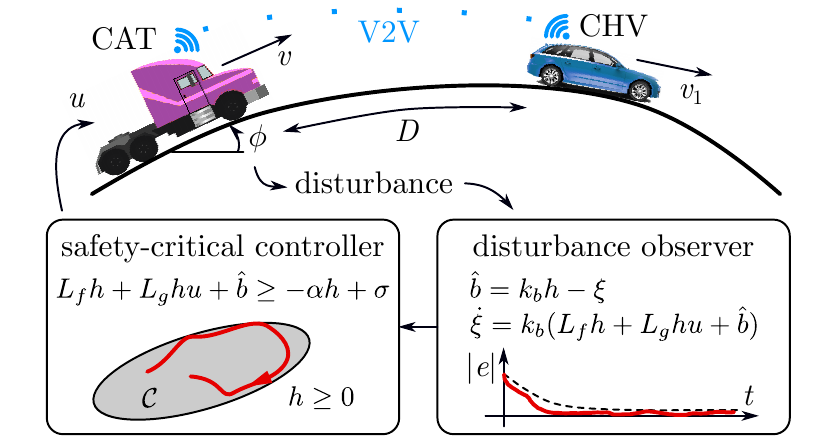}
	\caption{Illustration of the disturbance observer-based safety-critical control framework for an example where a connected automated truck (CAT) follows  a connected human-driven vehicle (CHV) without collision.}
	\label{fig:scheme}
	\vspace{-5mm}
\end{figure}

The paper is organized as follows.
Section~\ref{sec:preliminaries} gives a short summary about CBF theory.
Section~\ref{sec:observer} outlines the DOB and gives a bound for the observer error.
Section~\ref{sec:main} presents the proposed controller design and provides safety guarantees for appropriate observer and controller parameters.
Section~\ref{sec:discussion} gives further discussion on trade-offs for the parameter selection. Throughout the paper a \new{connected} cruise control problem is used as example 
for demonstration purposes.
\new{This practical example is selected for its simplicity to highlight the improvements achieved by the proposed method.
We also use real disturbance data to evaluate the method in a real-world scenario.}
Section~\ref{sec:conclusions} closes with conclusions.\

\section{Preliminaries}
\label{sec:preliminaries}

Consider the system:
\begin{equation}
    \label{eq:pre_system}
    \dot{x} = f(x) + g(x) u,
\end{equation}
with state ${x \in \R^n}$, input ${u\in\R^m}$, and locally Lipschitz continuous functions ${f:\R^{n}\to\R^n}$ and ${g:\R^{n}\to\R^{n\times m}}$.
Given a locally Lipschitz continuous controller ${k: \R^n \to \R^m}$, ${u=k(x)}$, we use the notation $x(t)$ for the unique solution of the corresponding closed-loop system with initial condition ${x(0) = x_0 \in \R^n}$, and we assume $x(t)$ exists for all ${t \geq 0}$. 

We define the notion of safety for the system \eqref{eq:pre_system} in accordance with the forward invariance of a \textit{safe set} $\C\subset \R^n$ given by a continuously differentiable function ${h : \R^n \to \R}$:
\begin{equation}
    \label{eq:C}
    \C \triangleq \left\{ x\in\R^n  ~|~ h(x)\geq 0 \right\}.
\end{equation}
The system \eqref{eq:pre_system} is said to be \textit{safe} with respect to the set $\C$ if it is forward invariant: $x_0 \in \C \implies x(t) \in \C$ for all ${t\geq0}$.

The function $h$ can be used to synthesize controllers that yield safe behavior. We say that $h$ is a \emph{Control Barrier Function (CBF)} \cite{ames2017control} if there exists\footnote{While we choose a constant $\alpha$ for simplicity, an extended class-$\cal{K}$ function, $\alpha : \R \to \R$, could be used more generally.} an ${\alpha>0}$ such that:
\begin{equation}
\label{eq:CBF}
     \sup_{u\in\R^m} \dot{h}(x,u) = \sup_{u\in\R^m} \left[ L_fh(x)+L_gh(x)u \right]>-\alpha h(x),
\end{equation}
where ${L_f h(x) \triangleq \derp{h(x)}{x} f(x)}$ and ${L_g h(x) \triangleq \derp{h(x)}{x} g(x)}$.
The set of CBF-based safe controllers is given as:
\begin{equation}
    \label{eq:K_CBF}
    K_{\rm CBF}(x) = \{ u\in\R^m ~|~ L_fh(x)+L_gh(x)u \geq -\alpha h(x) \}.
\end{equation}
One of the main results in \cite{ames2017control} proves that controllers taking values in $K_{\rm CBF}(x)$ for all ${x \in \R^n}$ lead to the safety of \eqref{eq:pre_system}.

\section{Disturbance Observer}
\label{sec:observer}

The safety guarantees of controllers in $K_{\rm CBF}$ may deteriorate in the presence of an
\new{uncertainty in the model.} 
In the rest of this paper, we consider the system:
\begin{equation}
    \label{eq:system}
    \dot{x} = f(x,r) + g(x) u + p(x,w),
\end{equation}
with state ${x \in \R^n}$, input ${u\in\R^m}$, reference signal ${r\in\R^l}$, disturbance ${w\in\R^q}$, and locally Lipschitz continuous functions
${f:\R^{n}\times\R^l\to\R^n}$, ${g:\R^{n}\to\R^{n\times m}}$ and ${p:\R^{n}\times\R^q\to\R^n}$.
The reference $r$ is assumed to be known,
thus it can be addressed by controllers ${k: \R^n \times \R^l \to \R^m}$, ${u=k(x,r)}$.
However, \new{$w$ and $p(x,w)$ are} unknown.

\new{The objective of the problem formulation is to quantify the effect of these uncertainties on safety.
Therefore}, we consider the time derivative of $h$ along the system:
\begin{equation}
    \label{eq:hdot}
    \dot{h}(x,u,r,w) =  L_fh(x,r) + L_gh(x)u + b(x,w),
\end{equation}
where
${b(x,w) \triangleq \derp{h(x)}{x} p(x,w)}$. If $b$ is negative at the safe set boundary (at ${h(x)=0}$), a controller in ${K_{\rm CBF}(x,r)}$ may fail to ensure that $\dot{h}$ is non-negative, which implies that the system would leave the safe set.

\begin{example} \label{ex:setup}
Consider the setup in Fig.~\ref{fig:scheme}, where a connected automated truck (CAT) is controlled to follow a connected human-driven vehicle (CHV).
Let $D$ be the distance between vehicles, $v$ and $v_1$ be the speeds of the CAT and CHV, and $u$ be the commanded acceleration of the CAT,
with dynamics:
\begin{align}
\begin{split}
    \dot{D} &= v_1-v, \\
    \dot{v} &= u - a(\phi) - c v^2,
\end{split} 
\label{eq:truck_system}
\end{align}
where $\phi$ is the time varying road grade,
$a(\phi) = g (\sin \phi + \gamma \cos \phi)$,
$g$ is the gravitational acceleration, 
$\gamma$ is the rolling resistance coefficient, and $c$ is the air drag coefficient.
We assume that the CHV's speed $v_1$ is available to the CAT through vehicle-to-vehicle (V2V) communication, hence it is a known reference, ${r=v_1}$. At the same time, we regard the road grade as an unknown disturbance, ${w=\phi}$.
By defining the state $x=[D,v]^\top$, system~\eqref{eq:truck_system} can be written as~\eqref{eq:system} with:
\begin{align}
\begin{split}
    f(x,r)\!=\!
    \begin{bmatrix}
    v_1-v \\
    -c v^2
    \end{bmatrix}\!, \;
    g(x)\!=\!
    \begin{bmatrix}
    0 \\ 1
    \end{bmatrix}\!, \;
    p(x,w)\!=\!
    \begin{bmatrix}
    0 \\ - a(\phi)
    \end{bmatrix}\!.
\end{split}
\label{eq:truck_functions}
\end{align}
This control problem is often referred to as connected cruise control \cite{orosz2016connected}.

\begin{figure}[t]
	\centering
	\includegraphics[trim=40 212 60 200,clip,width=0.95\linewidth]{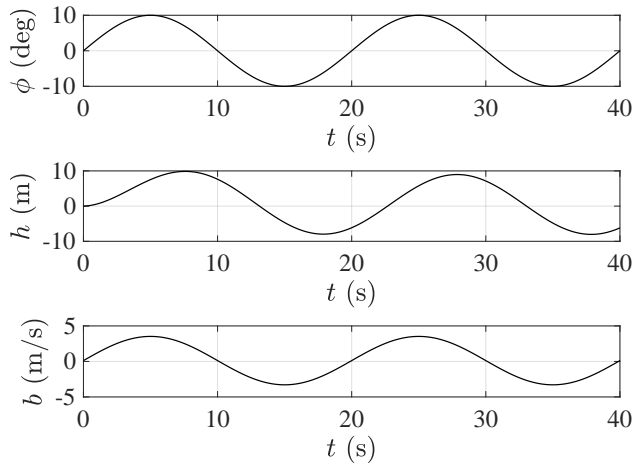}
	\caption{
	Simulations for Example~\ref{ex:setup}, with the time varying road grade that acts as disturbance (top), the evolution of the CBF $h$ (middle), and the effect $b$ of the disturbance on $\dot{h}$ (bottom).
    The controller that disregards the disturbance fails to maintain safety ($h$ goes negative).
    }
	\label{fig:Ex1}
\vspace{-3 mm}	
\end{figure}

To keep safe distance, we use a time-headway based CBF:
\begin{equation}
    \label{eq:truck_CBF}
    h(x) = D - D_{\rm sf} - T v,
\end{equation}
where ${D_{\rm sf}>0}$ is the safe stopping distance and ${T>0}$ is the safe time headway.
This yields ${L_fh(x,r)=v_1-v+Tcv^2}$ and ${L_gh(x)=-T}$.
We consider the controller ${u = k(x,r)}$,
\begin{multline}
    k(x,r) = - \frac{L_fh(x,r)+\alpha h(x)}{L_gh(x)} \\
    = \alpha \left( \kappa (D - D_{\rm sf}) - v \right) + \kappa (v_1 - v) + c v^2,
    \label{eq:ACC}
\end{multline}
with ${\kappa = 1/T}$, that is an element of ${K_{\rm CBF}(x,r)}$.
The controller disregards the road grade that acts as a disturbance.
Fig.~\ref{fig:Ex1} presents simulation results with parameters in Table~\ref{tab:param}, constant CHV speed ${v_1=v^*}$ and sinusoidal road grade:
\begin{equation}
    \label{eq:truck_road}
    \phi(t)=\Phi\sin(\omega t).
\end{equation}
The top and middle panels highlight that
safety is violated ($h$ becomes negative) due to the disturbance, whose effect ${b(x(t),w(t))=Ta(\Phi\sin(\omega t))}$ is plotted at the bottom.
\end{example}

We propose to use a disturbance observer to enforce safety robustly, under the following assumption.
\begin{assumption} \label{ass:bh}
Function $b(x(t),w(t))$ is Lipschitz continuous in $t$ over ${t \geq 0}$ with Lipschitz constant $b_h$.
\end{assumption}
\noindent \new{Note that Assumption~\ref{ass:bh} relaxes the assumption in \cite{das_dist:22} from differentiability to Lipschitz continuity.} If $b(x(t),w(t))$ is differentiable in $t$, $b_h$ is an upper bound on its derivative, ${|\der{}{t}b(x(t),w(t))| \leq b_h}$.

\vspace{1 mm}

To account for the unknown value of $b$, we utilize the high-gain first-order disturbance observer from \cite{das_dist:22}:
\begin{align}
    \label{eq:bhat}
    \hat{b}(x,\xi) &\triangleq k_b h(x) - \xi, \\
    \dot{\xi} &= \underbrace{k_b \left(L_fh(x,r) + L_gh(x)u + \hat{b}(x,\xi)\right)}_{f_\xi(x,u,r,\xi)}, \label{eq:xidot}
\end{align}
where ${\xi\in\R}$ is an auxiliary state and ${k_b>0}$ is the observer gain. 
By slight abuse of notation, we denote $b(x(t),w(t))$ and $\hat{b}(x(t),\xi(t))$ shortly as $b(t)$ and $\hat{b}(t)$.
We define the observer error ${e(t) \triangleq b(t) - \hat{b}(t)}$ with initial value ${e(0) = e_0 \in \R}$.
The error dynamics are characterized as follows.
\begin{lemma}
\label{lem:Vb}
Consider system \eqref{eq:system}, a continuously differentiable function $h$, function $b$ defined by \eqref{eq:hdot} with Lipschitz constant $b_h$, and the disturbance observer \eqref{eq:bhat}-\eqref{eq:xidot} with ${k_b>0}$.
The following bound holds for the error ${e(t) = b(t) - \hat{b}(t)}$:
\begin{equation}
    \label{eq:ebound}
    |e(t)| \leq \left(|e_0| - \frac{b_h}{k_b} \right) {\rm e}^{-k_b t} + \frac{b_h}{k_b}.
\end{equation}
\end{lemma}
\begin{proof}
Using \eqref{eq:hdot}, \eqref{eq:bhat} and \eqref{eq:xidot}, the observer dynamics read:
\vspace{-2 mm}
\begin{equation}
    \dot{\hat{b}} = k_b (b - \hat{b}).
\end{equation}
This is a linear dynamical system whose solution can be expressed by the convolution integral:
\begin{equation}
    \hat{b}(t) = \hat{b}(0) {\rm e}^{-k_b t} + \int_{0}^{t} {\rm e}^{-k_b(t-\theta)} k_b b(\theta) {\rm d}\theta.
\end{equation}

Hence, the evolution of the error is given by:
\begin{equation}
    e(t) = b(t) - b(0) {\rm e}^{-k_b t} + e_0 {\rm e}^{-k_b t}- \int_{0}^{t} {\rm e}^{-k_b(t-\theta)} k_b b(\theta) {\rm d}\theta,
    \label{eq:error_expression}
\end{equation}
where ${b(0) {\rm e}^{-k_b t}}$ was added and subtracted.
Via integration by parts, the following holds:
\begin{multline}
    b(t) - b(0) {\rm e}^{-k_b t} = \left[ b(\theta) {\rm e}^{-k_b(t-\theta)} \right]_{0}^{t} \\
    = \int_{0}^{t} {\rm e}^{-k_b(t-\theta)} k_b b(\theta) {\rm d}\theta + \int_{0}^{t} {\rm e}^{-k_b(t-\theta)} {\rm d}b(\theta),
    \label{eq:Lem1Proof_step}
\end{multline}
where a Stieltjes integral~\cite{riesz2012functional} is used to handle the potential non-differentiability of $b(\theta)$.
Substituting \eqref{eq:Lem1Proof_step}  into~\eqref{eq:error_expression} gives:
\begin{equation}
    e(t) = e_0 {\rm e}^{-k_b t}
    + \int_{0}^{t} {\rm e}^{-k_b(t-\theta)} {\rm d}b(\theta).
    \label{eq:error_integral}
\end{equation}

Due to the Lipschitz property of $b$ in Assumption~\ref{ass:bh}, the magnitude of the integral can be upper-bounded by:
\begin{equation}
    \left| \int_{0}^{t} {\rm e}^{-k_b(t-\theta)} {\rm d}b(\theta) \right|
    \leq \int_{0}^{t} {\rm e}^{-k_b(t-\theta)} b_h {\rm d}\theta
    = \frac{b_h}{k_b}\left( 1 - {\rm e}^{-k_b t} \right).
\end{equation}
With this, the absolute value of~\eqref{eq:error_integral} finally leads to~\eqref{eq:ebound}.
\end{proof}

\begin{table}[b]
    \centering
    \begin{tabular}{|rl|rl|rl|}
    \hline
     $g$           & 9.81 m/s$^2$  & $\gamma$  & 0.006   & $c$      & 0.000428 1/m \\ \hline
     $D_{\rm sf}$  & 5 m           & $T$       & 2 s     & $\alpha$ & 0.25 1/s \\ \hline
     $v^*$         & 20 m/s        & $\Phi$    & 10 deg  & $\omega$ & 0.05$\times$2$\pi$ rad/s.   \\ \hline
    \end{tabular}
    \caption{Parameters of the \new{connected automated} truck example.}
    \label{tab:param}
\end{table}

\begin{example} \label{ex:observer}
Consider the car-following setup in Example~\ref{ex:setup}.
For the sinusoidal road grade \eqref{eq:truck_road}, $b$ is differentiable with respect to $t$ such that
${|\dot{b}(t)| \leq b_h = T g \sqrt{1+\gamma^2} \Phi \omega \approx T g \Phi \omega}$. The evolution of $|\dot{b}|$ and $b_h$ are illustrated in the top panel of Fig.~\ref{fig:Ex2}. 
We employ the observer defined by~\eqref{eq:bhat}-\eqref{eq:xidot}, where:
\begin{align}
\begin{split}
    &\hat{b}(x,\xi) = k_b \left( D - D_{\rm sf} - T v \right) - \xi, \\
    &\dot{\xi} = k_b \big(v_1 \!-\! v \!+\! T c v^2 \!-\! T u \!+\! k_b \left( D \!-\! D_{\rm sf} \!-\! T v \right) \!-\! \xi \big).
\end{split}
\label{eq:truck_observer}
\end{align}
The bottom panel of Fig.~\ref{fig:Ex2} shows the observer error for ${|e_0|=5}$ and various $k_b$ values.
The observer error decreases with increasing $k_b$ and satisfies the bound \eqref{eq:ebound} for all cases.
\end{example}

\begin{figure}[t]
	\centering
	\includegraphics[trim=40 180 60 200,clip,width=.95\linewidth]{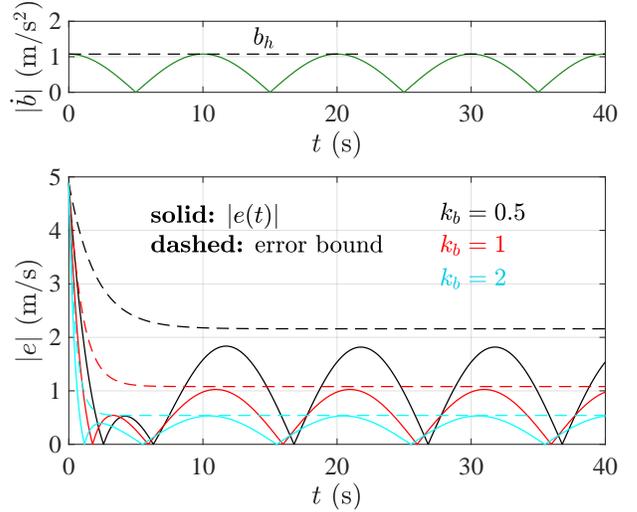}
	\vspace{-2 mm}
	\caption{
	Simulations for Example~\ref{ex:observer}, with the evolution of $|\dot{b}|$ and its upper bound $b_h$ (top) and the observer error for various observer gains $k_b$ (bottom).
	The observer error decreases with increasing $k_b$ and satisfies the error bound~\eqref{eq:ebound} in Lemma~\ref{lem:Vb}.
    }
	\label{fig:Ex2}
	\vspace{-4 mm}
\end{figure}

\begin{remark}
Lemma~\ref{lem:Vb} states the \textit{input-to-state stability}~\cite{sontag2008input} of the observer error dynamics around ${e=0}$. \new{The error bound~\eqref{eq:ebound} is stricter than the one presented in \cite{das_dist:22},}
and it consists of transient and steady-state parts; see Fig.~\ref{fig:Ex2}. 
The larger the observer gain $k_b $ is, the faster the transient decays and the narrower the steady-state error band is. 
\end{remark}

Next, we use the observed disturbance $\hat{b}$ to compensate for the unknown true disturbance $b$.
The observer error prevents ideal compensation. 
We introduce a modification to $K_{\rm CBF}$ to incorporate the disturbance observer into the controller and ensure safety.


\section{Main Result}
\label{sec:main}

We incorporate the disturbance observer into the CBF-based control design with the following modification of~\eqref{eq:K_CBF}:
\begin{align}
    \label{eq:Khat_CBF}
    \hat{K}_{\rm CBF}(x,r,\xi) = \{ & u\in\R^m  ~|~  L_fh(x,r) + L_gh(x)u \nonumber \\ 
    & + \hat{b}(x,\xi)  \geq -\alpha h(x) + \sigma  \}, 
\end{align}
where parameter ${\sigma>0}$ is inspired by the framework of \textit{input-to-state safety}~\cite{alan2021safe} to provide robustness against the observer error $e$.
To ensure that $\hat{K}_{\rm CBF}(x,r,\xi)$ is non-empty for any ${\xi \in \R}$, we assume that $h$ is a CBF with ${L_gh(x) \neq 0}$, ${\forall x \in \C}$.
The following theorem relates the controllers from $\hat{K}_{\rm CBF}$
to the safety of the disturbed system.
\begin{theorem}
\label{theo:main}
Consider system \eqref{eq:system}, CBF $h$ defining the set $\C$ as \eqref{eq:C}, function $b$ defined by \eqref{eq:hdot} with Lipschitz constant $b_h$, the disturbance observer \eqref{eq:bhat}-\eqref{eq:xidot} with ${k_b>0}$, and a Lipschitz continuous controller ${u = \hat{k}(x,r,\xi)\in \hat{K}_{\rm CBF}(x,r,\xi)}$.
\begin{itemize}[leftmargin=*]
    \item If ${\sigma \geq \max\{|e_0|,b_h/k_b\}}$, then $\C$ is rendered forward invariant, i.e., ${x_0 \in \C \implies x(t) \in \C}$.
    \item If ${\sigma \geq b_h/k_b}$ and ${k_b > \alpha}$, then ${x_0 \in \C_0 \cap \C \implies x(t) \in \C}$ with
    ${\C_0 = \left\{ x\in\R^n  ~|~ h(x)\geq (|e_0|-b_h/k_b)/(k_b-\alpha) \right\}}$.
\end{itemize}
\end{theorem}

\begin{proof}
By~\eqref{eq:Khat_CBF} and~\eqref{eq:ebound}, the time derivative~\eqref{eq:hdot} of $h$ satisfies:
\begin{align}
    \dot{h} &\geq - \alpha h + \sigma + e \\
     &\geq - \alpha h + \left( \frac{b_h}{k_b} - |e_0| \right) {\rm e}^{-k_b t} + \sigma - \frac{b_h}{k_b}
     \label{eq:safety_proof_1} \\
     &\geq - \alpha h + (E - |e_0|) {\rm e}^{-k_b t} + \sigma - E,
     \label{eq:safety_proof_2}
\end{align}
where
${E = \max\{|e_0|,b_h/k_b\} \geq |e_0|}$,
and function arguments were dropped for brevity.
If ${\sigma \geq E}$,~\eqref{eq:safety_proof_2} leads to ${\dot{h} \geq -\alpha h}$, and yields the first theorem statement.

To prove the second statement, consider the (unique) function ${y:\R_{\geq 0} \to \R}$ that satisfies:
\begin{align}
\begin{split}
    \dot{y} & = - \alpha y + \left( \frac{b_h}{k_b} - |e_0| \right) {\rm e}^{-k_b t} + \sigma - \frac{b_h}{k_b}, \\
    y(0) & = h(x_0).
\end{split}
\label{eq:y_IVP}
\end{align}
By applying the comparison lemma for~\eqref{eq:safety_proof_1} and~\eqref{eq:y_IVP}, we get ${h(x(t)) \geq y(t)}$, ${\forall t \geq 0}$, where $y(t)$ is obtained from~\eqref{eq:y_IVP} as:
\begin{multline}
    y(t) = \left( h(x_0) + \frac{b_h/k_b - |e_0|}{k_b - \alpha} \right) \left( {\rm e}^{-\alpha t} - {\rm e}^{-k_b t} \right) \\
    + h(x_0) {\rm e}^{-k_b t} + \frac{\sigma - b_h/k_b}{\alpha} \left( 1 - {\rm e}^{-\alpha t} \right).
\label{eq:comparison_y}
\end{multline}
Given ${k_b > \alpha>0}$, ${\sigma \geq b_h/k_b}$ and ${x_0 \in \C_0 \cap \C}$, each of the terms above are non-negative.
This leads to
${h(x(t)) \geq y(t) \geq 0}$, that is, ${x(t) \in \C}$, ${\forall t \geq 0}$.
\end{proof}

\begin{remark}
    The first statement of Theorem~\ref{theo:main} expresses that the set $\C$ can be made forward invariant for the disturbed system if parameter $\sigma$ is chosen to be large enough, such that it overcomes both the transient observer error (${\sigma \geq |e_0|}$) and the steady-state error bound (${\sigma \geq b_h/k_b}$) in~\eqref{eq:ebound}.
\end{remark}

\begin{remark}
    The second statement of Theorem~\ref{theo:main} addresses the case when parameter $\sigma$ overcomes the steady-state error (${\sigma \geq b_h/k_b}$) but not necessarily the transient error (potentially ${\sigma < |e_0|}$).
    Then, safety requires the initial state to satisfy ${x_0 \in \C_0 \cap \C}$.
    The larger the initial observer error $|e_0|$ is, the smaller $\C_0$ gets, which implies that the system must be located far inside the safe set initially.
    Additionally, safety requires large enough observer gain $k_b$, such that the convergence rate $k_b$ of the observer is larger than the rate $\alpha$ at which the system may approach the safe set boundary (${k_b > \alpha}$).
    A similar idea was used in \cite{Tamas__ROM_2022} to address safety when trajectories converge to those of a reduced order model.
\end{remark}

\begin{remark}
    One may also show the invariance of another set, ${\bar{\C} = \left\{ x\in\R^n  ~|~ \bar{h}(x)\geq 0 \right\}}$ with ${\bar{h}(x) = h(x) - (\sigma - E)/\alpha}$, since ${\dot{\bar{h}} \geq -\alpha \bar{h}}$ follows from~\eqref{eq:safety_proof_2}.
    Thus, even if parameter $\sigma$ is not large enough, ${\sigma < E}$, the system still evolves within a larger set ${\bar{\C} \supset \C}$ whose size is tuned by $\sigma$.
    As such, $\sigma$ provides robustness against disturbances.
    Meanwhile, ${\sigma > E}$ makes a smaller set ${\bar{\C} \subset \C}$ invariant, hence it may lead to conservativeness in the sense that trajectories may stay far inside the safe set $\C$.
    A similar trade-off was highlighted in \cite{alan2021safe,Alan__arxiv:22}, where the idea of tunable input-to-state safety with a variable $\sigma(h(x))$ was proposed to reduce conservativeness. 
\end{remark}

\begin{figure}[t]
	\centering
	\includegraphics[trim=7 5 15 10,clip,width=0.97\linewidth]{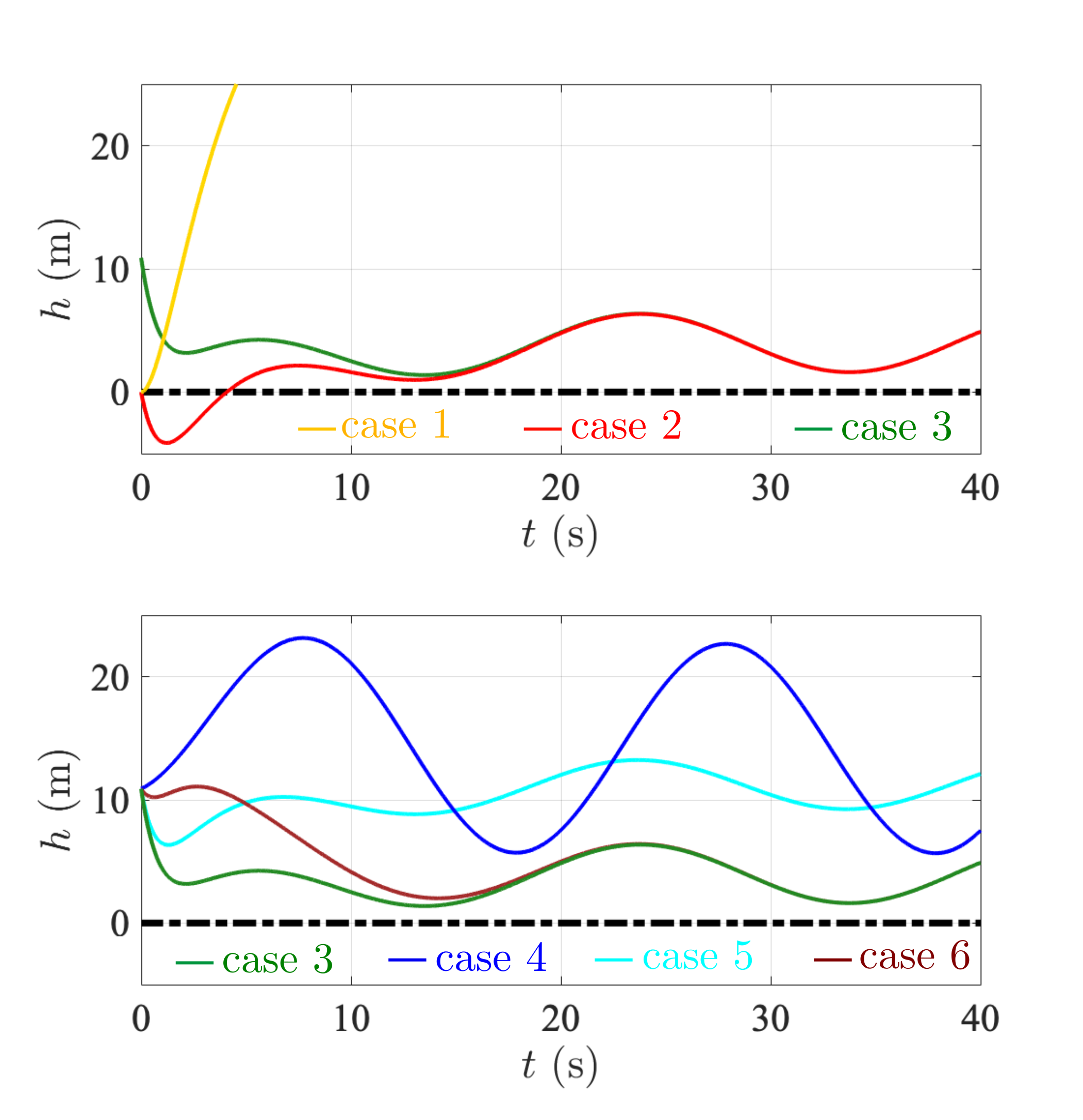}
	\caption{
	Simulations for Example~\ref{ex:controller}. 
	\new{(Top) Three cases are shown for the proposed method: safe but conservative case 1, unsafe case 2 and safe and not conservative case 3.
	(Bottom) Comparative results with controllers from \cite{jankovic2018robust}, \cite{Wang__DOB-CBF:22}, \cite{das_dist:22} as cases 4, 5 and 6, respectively, with respect to case 3.
	}
	}
	\label{fig:main}
	\vspace{-3 mm}
\end{figure}

\begin{example} \label{ex:controller}
Consider the car-following setup in Example~\ref{ex:setup}, the disturbance observer in Example~\ref{ex:observer}, and the controller:
\begin{multline}
    \label{eq:truck_k}
    \hat{k}(x,r,\xi) = - \frac{L_fh(x,r)+\alpha h(x) + \hat{b}(x,\xi)-\sigma}{L_gh(x)} \\
    = (\alpha + k_b) \left( \kappa (D - D_{\rm sf}) - v \right) + \kappa (v_1 - v) + c v^2 - \kappa (\xi + \sigma),
\end{multline}
with ${\kappa = 1/T}$ (cf. \eqref{eq:ACC}),
that is an element of ${\hat{K}_{\rm CBF}(x,r,\xi)}$. We evaluate the performance of the controller through simulations in three different cases based on Theorem~\ref{theo:main}:
\begin{enumerate}
    \item[{\color{black} \textbf{1)}}] ${\sigma=\max\{|e_0|,b_h/k_b\}}$ and $x_0$ is such that ${h(x_0)=0}$,
    \item[{\color{black} \textbf{2)}}] ${\sigma=b_h/k_b}$ and ${h(x_0)=0}$ (that is, ${x_0 \in \C}$ but ${x_0 \notin \C_0}$),
    \item[{\color{black} \textbf{3)}}] ${\sigma=b_h/k_b}$ and ${h(x_0)=(|e_0|-b_h/k_b)/(k_b-\alpha)>0}$ (that is, ${x_0 \in \C_0 \cap \C}$).
\end{enumerate}
We consider the sinusoidal road grade in \eqref{eq:truck_road}, use constant CHV speed profile ${v_1=v^*}$, pick $k_b$ such that ${b_h/k_b=1}$, and start from \new{${|e_0|=10}$ m/s} 
for all cases.

Simulation results are given in the top panel of Fig.~\ref{fig:main}.
Case~1 satisfies the condition in the first point of Theorem~\ref{theo:main}, hence it results in safety. Since $|e_0|$ is large, an equivalently large $\sigma$ yields conservativeness by pushing the trajectory farther inside the safe set. 
Case 2 and case 3 refer to the second point in Theorem~\ref{theo:main}, where the former fails to satisfy the required initial condition and the latter starts within $\C_0$. As such, case 2 leads to safety violation during the transient due to the large $|e_0|$. Case 3, on the other hand, keeps the system safe thanks to starting inside ${\C_0 \subset \C}$,
cf.~\eqref{eq:error_integral}.
\new{Additionally, we implement three controllers from the literature, see bottom panel of Fig.~\ref{fig:main}.
Case 4 shows a worst-case approach from  \cite{jankovic2018robust} with ${\|p(x,w)\|_\infty\leq\overline{p}}$, that yields conservative results.
Case 5 presents a disturbance observer approach from \cite{Wang__DOB-CBF:22} for the disturbance $d=\sin \phi$, which alleviates the conservativeness of the worst-case approach, yet overcompensates for the steady state error due to large initial observer error.
Case 6 denotes the approach of~\cite{das_dist:22}, that directly cancels the transient error using the error bound, therefore results in conservative behavior during the initial transient with respect to case 3.
}
\end{example}

\section{Discussion}
\label{sec:discussion}
Choosing a larger observer gain $k_b$ attains stricter observer error bounds, and consequently a less conservative controller by indulging a smaller robustness parameter $\sigma$.
However, large gains may lead to instability in the presence of unmodeled dynamics.
Next, we demonstrate this for an unmodeled input time delay.
We employ linear stability analysis to investigate the limitations of controllers in ${\hat{K}_{\rm CBF}}$ due to the delay.
Finally, we utilize real road grade and CHV speed data to assess the controller in the example using simulations.

Consider the system with a constant input time delay ${\tau>0}$ representing actuator dynamics:
\begin{equation}
    \label{eq:system_eta}
    \dot{x}(t) = f(x(t),r(t)) + g(x(t)) u(t-\tau) + p(x(t),w(t)),
\end{equation}
(cf. \eqref{eq:system}) and a controller ${u = \hat{k}(x,r,\xi)\in\hat{K}_{\rm CBF}(x,r,\xi)}$.
By defining ${z(t)\triangleq[x(t),\xi(t)]^\top\in\R^{n+1}}$, ${z_\tau(t)\triangleq z(t-\tau)}$ and ${r_\tau(t)\triangleq r(t-\tau)}$, we obtain the closed-loop dynamics: 
\begin{align}
    \label{eq:closedloop_z_eta}
    \dot{z} = F(z,z_\tau,r,r_\tau) + p_z(z,w).
\end{align}
with ${F(z,z_\tau,r,r_\tau) = f_z(z,r) + g_z(z) k_z(z_\tau,r_\tau)}$ and:
\begin{align}
\begin{split}
    f_z(z,r) &= \begin{bmatrix}
    f(x,r) \\ f_\xi(x,\hat{k}(x,r,\xi),r,\xi)
    \end{bmatrix}\!, \;
    g_z(z) = \begin{bmatrix}
    g(x) \\ \mathbf{0}_m
    \end{bmatrix}\!, \\
    k_z(z,r) &= \hat{k}(x,r,\xi), \quad
    p_z(z,w) = \begin{bmatrix}
    p(x,w) \\ \mathbf{0}_q
    \end{bmatrix}\!,
\end{split}
\end{align}
where $f_\xi$ is as defined in \eqref{eq:xidot}, while $\mathbf{0}_m$ and $\mathbf{0}_q$ are zero column vectors with dimensions $m$ and $q$.

To conduct linear stability analysis, we assume that functions $F$ and $p_z$ are differentiable at an equilibrium ${z\equiv z^*}$, ${r\equiv r^*}$ and ${w\equiv \mathbf{0}_q}$.
Note that this assumption was not required for Theorem~\ref{theo:main}.
Defining ${\tilde{z}\triangleq z-z^*}$, ${\tilde{z}_\tau\triangleq z_\tau-z^*}$, ${\tilde{r}\triangleq r-r^*}$ and ${\tilde{r}_\tau\triangleq r_\tau-r^*}$,
the linearized dynamics are:
\begin{align}
    \label{eq:closedloop_z_lin}
    \dot{\tilde{z}} = A \tilde{z} + A_\tau \tilde{z}_\tau + B_w w + B_r \tilde{r} + B_{r_\tau} \tilde{r}_\tau,
\end{align}
where the coefficient matrices read:
\begin{align}
\begin{split}
    A &= \derp{F}{z}\bigg\rvert_{z^*,r^*} + \derp{p_z}{z}\bigg\rvert_{z^*,\mathbf{0}_q}\!, \;
    A_\tau = \derp{F}{z_\tau}\bigg\rvert_{z^*,r^*}\!, \\
    B_w &= \derp{p_z}{w}\bigg\rvert_{z^*,\mathbf{0}_q}\!, \;
    B_r = \derp{F}{r}\bigg\rvert_{z^*,r^*}\!, \;
    B_{r_\tau} = \derp{F}{r_\tau}\bigg\rvert_{z^*,r^*}\!,
\end{split}
\end{align}
evaluated at ${z=z_\tau=z^*}$, ${r=r_\tau=r^*}$ and ${w=\mathbf{0}_q}$.

System \eqref{eq:closedloop_z_lin} is associated with the characteristic function:
\begin{align}
    H(s) = \det{\left(sI-A-A_\tau e^{-s\tau}\right)}
    \label{eq:charfun}
\end{align}
and the characteristic equation ${H(s)=0}$, 
with $I$ being the identity matrix.
For stability, all the infinitely many roots of   this equation must have negative real parts \cite{insperger2011semi}.
At the stability limit, ${H(j\Omega)=0}$ holds for some ${\Omega\geq0}$.
This leads to two algebraic equations after separating real and imaginary parts, which can be solved for parameters of interest like $\alpha$ and $k_b$. The solution yields the stability boundaries that can be plotted as \textit{stability charts} in the space of parameters; see \cite{orosz2016connected} for details. 
We present stability charts for an example.

\begin{figure*}[t]
	\centering
	\includegraphics[trim=0 0 0 0,clip,width=1\linewidth]{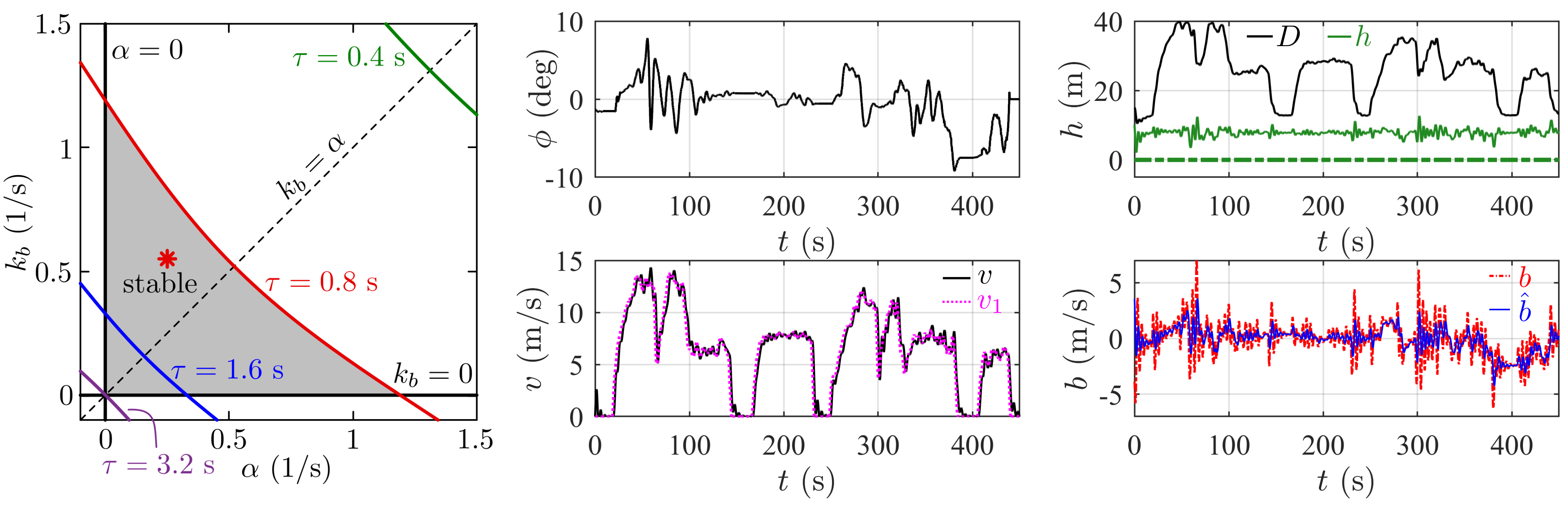}
	\vspace{-3 mm}
	\caption{(Left) Stability charts for Example~\ref{ex:delay}.
	Gray shading denotes the stable region for ${\tau = 0.8}$ s, and red asterisk indicates the parameters selected for simulations.
	The stable region shrinks with increasing delay, and delay prevents selecting an arbitrarily large $k_b$ without instability. (Middle and right) Simulations for Example 4 with real road grade and CHV speed data. The proposed controller maintains safety despite the delay, and the closed-loop system is guaranteed to be linearly stable by careful parameter selection.}
	\vspace{-3 mm}
	\label{fig:chart}
\end{figure*}

\begin{example} \label{ex:delay}
Consider the car-following setup in Example~\ref{ex:setup}, the disturbance observer in Example~\ref{ex:observer} and the controller~\eqref{eq:truck_k} in Example~\ref{ex:controller}.
With an input time delay ${\tau > 0}$ representing computation and communication lags as well as the time required for the CAT to realize brake and engine commands corresponding to the input $u$, we have:
\begin{align}
\begin{split}
    \dot{D}(t) &= v_1(t)-v(t), \\
    \dot{v}(t) &= u(t-\tau) - a(\phi(t)) - c v(t)^2,
\end{split}
\label{eq:truck_system_delay}
\end{align}
that is of form~\eqref{eq:system_eta} with~\eqref{eq:truck_functions}.
The characteristic function 
\begin{align}
\notag
    H(s) = & \left( s^3+ 2 c v^* s^2 \right){\rm e}^{s \tau} + \left( \alpha + k_b + \kappa \right)s^2 \\ 
    &+ \left( (\alpha+k_b) \kappa + \alpha k_b \right)s + \alpha k_b \kappa
    \label{eq:truck_H(s)}
\end{align}
\new{does not contain} $\sigma$, only $\alpha$ and $k_b$.

We calculate the linear stability boundaries as curves parameterized by ${\Omega \geq 0}$, by solving ${H(j\Omega)=0}$ for $\alpha$ and $k_b$.
We plot the boundaries in the ${(\alpha,k_b)}$ parameter space for different delay values, that yields the stability charts in Fig.~\ref{fig:chart}.
Note that the boundaries ${\alpha = 0}$ and ${k_b = 0}$ correspond to ${\Omega = 0}$.
Fig.~\ref{fig:chart} highlights that the observer gain $k_b$ cannot be selected arbitrarily large without instability, and that the stable region shrinks as the delay increases.
At a critical delay $\tau_{\rm cr}$ the stability boundary runs through the origin, and the stable region disappears for ${\tau > \tau_{\rm cr}}$.
The critical delay $\tau_{\rm cr}$ can be found by solving ${H(j\Omega)=0}$ with
${\alpha = 0}$, ${k_b = 0}$
for $\tau$ and $\Omega$, that leads to ${\tau_{\rm cr} = \arcsin(\Omega_{\rm cr}/\kappa)/\Omega_{\rm cr} \approx \pi/(2\kappa)}$ with ${\Omega_{\rm cr} = \sqrt{\kappa^2 - 4 c^2 {v^{*}}^2} \approx \kappa}$.
Given the parameters in Table~\ref{tab:param}, all ${(\alpha,k_b)}$ pairs lead to instability for ${\tau=3.2\ {\rm s}>\tau_{\rm cr}}$.
From now on, we let ${\tau=0.8}$~s, and we choose ${\alpha=0.25}$~1/s and ${k_b=0.55}$~1/s as highlighted by the red asterisk.

Next, we evaluate the robustness of the controller against the input delay by simulations.
We use real data for the CHV's speed profile and for the road grade \cite{Chaozhe__ACC:20} as depicted in Fig.~\ref{fig:chart} (middle).
Notice that around ${t=390}$~s the CHV brakes hard while traveling on steep downhill, leading to a particularly safety-critical situation.
To simulate the CAT's motion, we use the same $|e_0|$ and $b_h$ values as in Example 2, and we invoke the case 3 in Example 3 with ${\sigma=b_h/k_b=1.96}$~m/s and ${h(x_0)=(|e_0|-b_h/k_b)/(k_b-\alpha)=10.1}$~m.
This setup is guaranteed to be safe in the absence of the delay based on Theorem~\ref{theo:main}.
With delay, the controller still maintains safety throughout the run even at the most critical moment at ${t=390}$~s thanks to the disturbance observer; see Fig.~\ref{fig:chart} (right). The disturbance observer $\hat{b}$ tracks the unknown effect of the model mismatch on safety, visualized as $b$ in Fig.~\ref{fig:chart} (right). Meanwhile, stability is guaranteed as parameters were chosen based on the stability chart in Fig.~\ref{fig:chart}.

\end{example}

\section{Conclusions}
\label{sec:conclusions}

This paper addressed the safety-critical control of systems with model uncertainties.
We \new{used a disturbance observer technique to estimate the effect of the uncertainty} on the safety, and we incorporated \new{the observer} into the control design to provide robust safety guarantees by control barrier functions.
We \new{gave conditions on controller parameters that} lead to provable safety, and we discussed the \new{practical} limitations on choosing high parameters.
We demonstrated the efficacy of the proposed \new{method} using numerical simulations \new{for a connected} cruise control system using real data.

\new{Future work includes implementing the proposed framework to other applications, and adding a tunability feature from \cite{alan2021safe} for less conservative results under significant permanent error bounds. Furthermore, enforcing robust safety under multiplicative uncertainties (such as uncertainties in the control matrix $g(x)$) is another topic for future study.
}

\bibliographystyle{IEEEtran}
\bibliography{Bib/Alan_bib}

\end{document}